\documentclass[runningheads,envcountsame,11pt]{llncs}

\usepackage{tikz}
\usetikzlibrary{snakes}
\usepackage{xspace}
\usepackage{comment}

\usepackage{makeidx}  
\usepackage{amsmath,amstext,amssymb}

\begin{document}

\newcommand{\crbds}{{\sc{Colourful Red-Blue Dominating Set}}\xspace}
\newcommand{\crbdsshort}{{\sc{cRBDS}}\xspace}
\newcommand{\vddeg}{\textsc{Maximum $d$-degenerate Induced Subgraph}\xspace}
\newcommand{\vddegpar}{$d${\sc{-degenerate Vertex Deletion}}\xspace}
\newcommand{\tddegpar}{$2${\sc{-degenerate Vertex Deletion}}\xspace}
\renewcommand{\subset}{{\subseteq}}
\newcommand{\Oo}{\mathcal{O}}
\newcommand{\Ohstar}{\Oo^\star}
\newcommand{\eps}{\varepsilon}
\newcommand{\alg}{\mathcal{A}}
\newcommand{\degbound}{\kappa}
\newcommand{\EVratio}{\lambda}
\newcommand{\stopratio}{\alpha}
\newcommand{\palantratio}{c}
\newcommand{\maybeqed}{\qed}
\newcommand{\NP}{NP}

\spnewtheorem{algrule}{Rule}{\bfseries}{\rm}

\newcommand{\defproblem}[4]{
  \vspace{2mm}
\noindent\fbox{
  \begin{minipage}{0.96\textwidth}
  \begin{tabular*}{0.96\textwidth}{@{\extracolsep{\fill}}lr} #1 & {\bf{Parameter:}} #3 \\ \end{tabular*}
  {\bf{Input:}} #2  \\
  {\bf{Question:}} #4
  \end{minipage}
  }
  \vspace{2mm}
}
\newcommand{\defnoparamproblem}[3]{
  \vspace{2mm}
\noindent\fbox{
  \begin{minipage}{0.96\textwidth}
  #1
  {\bf{Input:}} #2  \\
  {\bf{Question:}} #3
  \end{minipage}
  }
  \vspace{2mm}
}

\title{Finding a maximum induced degenerate subgraph faster than $2^n$}

\author{
  Marcin Pilipczuk\inst{1}\thanks{Partially supported by NCN grant N206567140 and Foundation for Polish Science.}
  \and
  Micha\l{} Pilipczuk\inst{2}\thanks{Partially supported by European Research Council (ERC) Grant ``Rigorous Theory of Preprocessing'', reference 267959.}
}
\authorrunning{Pilipczuk$\times 2$}
\titlerunning{Finding a maximum degenerate\ldots}

\institute{Institute of Informatics, University of Warsaw, Poland\\
  \email{malcin@mimuw.edu.pl}
  \and
  Department of Informatics, University of Bergen, Norway\\
  \email{michal.pilipczuk@ii.uib.no}
}

\maketitle

\begin{abstract}
In this paper we study the problem of finding a maximum induced $d$-degenerate subgraph in a given $n$-vertex graph from the point of view of exact algorithms. We show that for any fixed $d$ one can find a maximum induced $d$-degenerate subgraph in randomized $(2-\eps_d)^nn^{\Oo(1)}$ time, for some constant $\eps_d>0$ depending only on $d$. Moreover, our algorithm can be used to sample inclusion-wise maximal induced $d$-degenerate subgraphs in such a manner that every such subgraph is output with probability at least $(2-\eps_d)^{-n}$; hence, we prove that their number is bounded by $(2-\eps_d)^n$.
\end{abstract}

\section{Introduction}



The theory of exact computations studies the design of algorithms for \NP-hard problems that compute the answer optimally, however using possibly exponential time. The goal is to limit the exponential blow-up in the best possible running-time guarantee. For some problems, like {\sc Independent Set}~\cite{fgk:m-c-jacm}, {\sc Dominating Set}~\cite{fgk:m-c-jacm,rooij:domsetesa09}, and
{\sc Bandwidth}~\cite{nasz-tcs} the research concentrates on achieving better and better constants in the bases of exponents. However, for many important computational tasks designing even a routine faster than trivial brute-force solution or straightforward dynamic program is a challenging combinatorial question; the answer to this question can provide valuable insight into the structure of the problem. Perhaps the most prominent among recent developments in breaking trivial barriers is the algorithm for {\sc Hamiltonian Cycle} of Bj\"orklund~\cite{bjorklund-focs}, but a lot of effort is put also into less fundamental problems, like \textsc{Maximum Induced Planar Graph}~\cite{planar} or a scheduling problem $1|\textrm{prec}|\sum C_i$~\cite{sched}, among many
others \cite{capdom,irredundant,twoconn,exact-sfvs,razgon-mif,exact-fvs}. However, many natural and well-studied problems still lack exact algorithms faster than the trivial ones; the most important examples are {\sc TSP}, {\sc Permanent}, {\sc Set Cover}, {\sc \#Hamiltonian Cycles} and {\sc SAT}. In particular, hardness of {\sc SAT} is the starting point for the {\emph{Strong Exponential Time Hypothesis}} of Impagliazzo and Paturi~\cite{seth,seth2}, which is used as an argument that other problems are hard as well~\cite{cut-and-count,treewidth-lower,patrascu,sodaomc}. 

A group of tasks we are particularly interested in in this paper are the problems that ask for a maximum size induced subgraph belonging to some class~$\Pi$. If belonging to $\Pi$ can be recognized in polynomial time, then we have an obvious brute-force solution working in $2^n n^{\Oo(1)}$ time that iterates through all the subsets of vertices checking which of them induce subgraphs belonging to $\Pi$. Note that the classical {\sc{Independent Set}} problem can be formulated in this manner for $\Pi$ being the class of edgeless graphs, while if $\Pi$ is the class of forests then we arrive at the {\sc{Maximum Induced Forest}}, which is dual to {\sc{Feedback Vertex Set}}. For both these problems algorithms with running time of form $(2-\eps)^n$ for some $\eps>0$ are known~\cite{fgk:m-c-jacm,razgon-mif,exact-fvs}. The list of problems admitting algorithms with similar complexities includes also $\Pi$ being the classes of regular graphs~\cite{regular}, graphs of small treewidth~\cite{small-tw}, planar graphs~\cite{planar}, $2$- or $3$-colourable graphs~\cite{colourable}, bicliques~\cite{bicliques} or graphs excluding a forbidden subgraph~\cite{serge}.

The starting point of our work is the question raised by Fomin et al. in~\cite{planar}. Having obtained an algorithm finding a maximum induced planar graph in time $\Oo(1.7347^n)$, they ask whether their result can be extended to graphs of bounded genus or even to $H$-minor-free graphs for fixed $H$. Note that all these graph classes are hereditary and consist of {\emph{sparse}} graphs, i.e., graphs with the number of edges bounded linearly in the number of vertices. Moreover, for other hereditary sparse classes, such as graphs of bounded treewidth, algorithms with running time $(2-\eps)^n$ for some $\eps>0$ are also known~\cite{small-tw}. Therefore, it is tempting to ask whether the sparseness of the graph class can be used to break the $2^n$ barrier in a more general manner.

In order to formalize this question we study the problem of finding a maximum induced $d$-degenerate graph. Recall that a graph is called $d$-degenerate if each of its subgraphs contains a vertex of degree at most $d$. Every hereditary class of graphs with a number of edges bounded linearly in the number of vertices is $d$-degenerate for some $d$; for example, planar graphs are $5$-degenerate, graphs excluding $K_r$ as a minor are $\Oo(r \sqrt{\log r})$-degenerate, while the class of forests is equivalent to the class of $1$-degenerate graphs. However, $d$-degeneracy does not impose any topological constraints; to see this, note that one can turn any graph into a $2$-degenerate graph by subdividing every edge. Hence, considering a problem on the class of $d$-degenerate graphs can be useful to examine whether it is just sparseness that makes it more tractable, or one has to add additional restrictions of topological nature~\cite{bimber}.

\paragraph{Our results and techniques.}
We make a step towards understanding the complexity of finding a maximum 
induced subgraph from a sparse graph class by breaking the $2^n$-barrier for the
problem of finding maximum induced $d$-degenerate subgraph.
The main result of this paper is the following algorithmic theorem.

\begin{theorem}\label{thm:alg}
For any integer $d \geq 1$ there exists a constant $\eps_d > 0$ and a polynomial-time
randomized algorithm $\alg_d$, which given an $n$-vertex graph $G$ either reports an error, or outputs a subset of vertices inducing a $d$-degenerate subgraph. Moreover, for every inclusion-wise maximal induced $d$-degenerate subgraph, let $X$ be its vertex set, the probability that $\alg_d$ outputs $X$ is at least $(2-\eps_d)^{-n}$.
\end{theorem}

Let $X_0$ be a set of vertices inducing a maximum $d$-degenerate subgraph. If we run the algorithm $(2-\eps_d)^n$ times, we know that with probability at least $1/2$ in one of the runs the set $X_0$ will be found. Hence, outputting the maximum size set among those found by the runs gives the following corollary.

\begin{corollary}\label{cor:max}
There exists a randomized algorithm which, given an $n$-vertex graph $G$, in $(2-\eps_d)^n n^{\Oo(1)}$ time outputs a set $X\subseteq V(G)$ inducing a $d$-degenerate graph. Moreover, $X$ is maximum with probability at least~$\frac{1}{2}$.
\end{corollary}

As the total probability that $\alg_d$ outputs some set of vertices is bounded by~$1$, we obtain also the following corollary.

\begin{corollary}\label{cor:bound}
For any integer $d \geq 1$ there exists a constant $\eps_d>0$
such that any $n$-vertex graph contains at most $(2-\eps_d)^n$ inclusion-wise
maximal induced $d$-degenerate subgraphs.
\end{corollary}

Let us elaborate briefly on the idea behind the algorithm of Theorem \ref{thm:alg}.
Assume first that $G$ has large average degree, i.e., $|E(G)| > \EVratio d |V(G)|$ 
for some large constant $\EVratio$. As $d$-degenerate graphs are sparse, i.e., the number 
of edges is less than $d$ times the number of vertices, it follows that for any set $X$ inducing 
a $d$-degenerate graph $G[X]$, only a tiny fraction of edges inside $G$ are in fact inside $G[X]$.
Hence, an edge $uv$ chosen uniformly at random can be assumed with high probability to have at least 
one endpoint outside $X$. We can further choose at random, with probabilities $1/3$ each, one 
of the following decisions: $u\in X$, $v \notin X$ or $u \notin X$, $v \in X$, or $u,v \notin X$. 
In this manner we fix the status of two vertices of $G$ and, if $\EVratio > 4$,
the probability that the guess is correct is larger than $1/4$. If this randomized step cannot be applied, we know 
that the average degree in $G$ is at most $\EVratio d$ and we can apply more standard branching arguments on vertices 
of low degrees.

Our algorithm is a polynomial-time routine that outputs an induced $d$-degenerate graph by guessing assignment of consecutive vertices with probabilities slightly better than $1/2$. We would like to remark that all but one of the ingredients of the algorithm can be turned into standard, deterministic branching steps. The only truly randomized part is the aforementioned random choice of an edge to perform a guess with enhanced success probability. However, to ease the presentation we choose to present the whole algorithm in a randomized fashion by expressing classical branchings as random choices of the branch.

\paragraph{Organization.}
In Section \ref{sec:prelims} we settle notation and give preliminary results on degenerate graphs. Section \ref{sec:alg} contains the proof of Theorem \ref{thm:alg}.
Section \ref{sec:conc} concludes the paper.
\section{Preliminaries}\label{sec:prelims}

\paragraph{Notation.} We use standard graph notation.
For a graph~$G$, by~$V(G)$ and~$E(G)$ we denote its vertex and edge sets, respectively.
For~$v \in V(G)$, its neighborhood~$N_G(v)$ is defined as~$N_G(v) = \{u: uv\in E(G)\}$.
For a set~$X \subseteq V(G)$ by~$G[X]$ we denote the subgraph of~$G$ induced by~$X$.
For a set~$X$ of vertices or edges of~$G$, by~$G \setminus X$ we denote the graph
with the vertices or edges of~$X$ removed; in case of vertex removal, we remove
also all the incident edges.

\paragraph{Degenerate graphs.}
For an integer $d \geq 0$, we say that a graph $G$ is $d$-degenerate
if every subgraph (equivalently, every induced subgraph) of $G$ contains a vertex
of degree at most $d$. Clearly, the class of $d$-degenerate graphs is closed under taking
both subgraphs and induced subgraphs.
Note that $0$-degenerate graphs are independent sets, and the class of $1$-degenerate graphs
is exactly the class of forests.
All planar graphs are $5$-degenerate;
moreover, every $K_r$-minor-free graph (in particular, any
$H$-minor-free graph for $|V(H)| = r$) is $\Oo(r \sqrt{\log r})$-degenerate \cite{minorddeg1,minorddeg2,minorddeg3}. 

The following simple proposition shows that the notion of $d$-degeneracy admits greedy arguments.

\begin{proposition}\label{prop:greedy}
Let $G$ be a graph and $v$ be a vertex of degree at most $d$ in $G$. Then $G$ is $d$-degenerate if and only if $G\setminus v$ is.
\end{proposition}
\begin{proof}
As $G\setminus v$ is a subgraph of $G$, then $d$-degeneracy of $G$ implies $d$-degeneracy of $G\setminus v$. Hence, we only need to justify that if $G\setminus v$ is $d$-degenerate, then so does $G$. Take any $X\subseteq V(G)$. If $v\in X$, then the degree of $v$ in $G[X]$ is at most its degree in $G$, hence it is at most $d$. However, if $v\notin X$ then $G[X]$ is a subgraph of $G\setminus v$ and $G[X]$ contains a vertex of degree at most $d$ as well. As $X$ was chosen arbitrarily, the claim follows.\maybeqed
\end{proof}

Proposition~\ref{prop:greedy} ensures that one can test $d$-degeneracy of a graph by in turn finding a vertex of degree at most $d$, which needs to exist due to the definition, and deleting it. If in this manner we can remove all the vertices of the graph, it is clearly $d$-degenerate. Otherwise we end up with an induced subgraph with minimum degree at least $d+1$, which is a sufficient proof that the graph is not $d$-degenerate. Note that this procedure can be implemented in polynomial time. As during each deletion we remove at most $d$ edges from the graph, the following proposition is straightforward.

\begin{proposition}\label{prop:edgebound}
Any $n$-vertex $d$-degenerate graph has at most $dn$ edges.
\end{proposition}
\section{The algorithm}\label{sec:alg}

In this section we prove Theorem \ref{thm:alg}.
Let us fix $d \geq 1$, an $n$-vertex graph $G$ and an inclusion-wise maximal set $X \subseteq V(G)$ inducing a $d$-degenerate graph.

The behaviour of the algorithm depends on a few constants that may depend on $d$
and whose values influence the final success probability.
At the end of this section we propose precise values of these constants
and respective values of $\eps_d$ for $1\leq d\leq 6$.
However, as the values of $\eps_d$ are really tiny even for small $d$,
when describing the algorithm we prefer to 
introduce these constants symbolically,
and only argue that there exists their evaluation that leads to a $(2-\eps_d)^{-n}$ lower bound on the probability of successfully sampling $X$.

The algorithm maintains two disjoint sets $A,Z \subseteq V(G)$, consisting of vertices about which we have already made some assumptions: we seek for the set $X$ that contains $A$ and is disjoint from $Z$. Let $Q = V(G) \setminus (A \cup Z)$ be the set of the remaining vertices, whose assignment is not yet decided.

We start with $A = Z = \emptyset$.
The description of the algorithm consists of a sequence of {\em{rules}}; at each point, the lowest-numbered applicable rule is used.
When applying a rule we assign some vertices of $Q$ to the set $A$ or $Z$, depending on some random decision.
We say that an application of a rule is {\em{correct}} if, assuming that before the application we have $A \subseteq X$ and $Z \cap X = \emptyset$,
the vertices assigned to $A$ belong to $X$, and the vertices assigned to $Z$ belong to $V(G) \setminus X$. 
In other words, a correct application assigns the vertices consistently with the fixed solution $X$.

We start with the randomized rule that is triggered when the graph is dense. Observe that, since $G[X]$ is $d$-degenerate, $G[X \cap Q]$ is $d$-degenerate as well
and, by Proposition \ref{prop:edgebound}, contains less than $d|X \cap Q|$ edges.
Thus, if $|E(G[Q])|/|Q|$ is significantly larger than $d$, then only a tiny fraction of the edges of $G[Q]$
are present in $G[X]$. Hence, an overwhelming fraction of edges of $G[Q]$ has at least one of the endpoints outside $X$, so having sampled an edge of $G[Q]$ uniformly at random with high probability we may assume that there are only three possibilities of the behaviour of its endpoints, instead of four. This observation leads to the following rule. Let $\EVratio > 4$ be a constant.
\begin{algrule}\label{rule:randedge}
If $|E(G[Q])| \geq \EVratio d|Q|$, then:
\begin{enumerate}
\item choose an edge $uv \in E(G[Q])$ uniformly at random;
\item with probability $1/3$ each, make one of the following decisions: either assign
$u$ to $A$ and $v$ to $Z$, or assign $u$ to $Z$ and $v$ to $A$, or assign both $u$ and $v$
to $Z$.
\end{enumerate}
\end{algrule}
\begin{lemma}\label{lem:randedge}
Assume that $A \subseteq X$ and $Z \cap X = \emptyset$ before Rule \ref{rule:randedge} is applied. Then the application of Rule \ref{rule:randedge} is correct with probability
at least $\frac{\EVratio - 1}{3\EVratio}$.
\end{lemma}
\begin{proof}
As $|E(G[Q])| \geq \EVratio d|Q|$, but $|E(G[X \cap Q])| \leq d|X \cap Q| \leq d|Q|$
by Proposition \ref{prop:edgebound},
the probability that $uv \notin E(G[X])$
is at least $\frac{\EVratio - 1}{\EVratio}$.
Conditional on the assumption $uv \notin E(G[X])$, in the second step of
Rule~\ref{rule:randedge} we make a correct decision with probability $1/3$.
This concludes the proof.\maybeqed
\end{proof}
Note that the bound $\frac{\EVratio-1}{3\EVratio}$ is larger than $1/4$
for $\EVratio > 4$.

Equipped with Rule \ref{rule:randedge}, we may focus on the case when
$G[Q]$ has small average degree. Let us introduce a constant $\degbound > 2\EVratio$ and let $S\subseteq Q$ be the set of vertices having degree less than $\degbound d$ in $G[Q]$. If Rule \ref{rule:randedge} is not applicable, then $|E(G[Q])| < \EVratio d|Q|$. Hence we can infer that $|S|\geq \frac{\degbound-2\EVratio}{\degbound}|Q|$, as otherwise by just counting the degrees of vertices in $Q\setminus S$ we could find at least $\frac{1}{2}\cdot \frac{2\EVratio}{\degbound}|Q| \cdot \degbound d=\EVratio d|Q|$ edges in $G[Q]$. Consider any $v\in S$. Such a vertex $v$ may be of two types: it either has at most
$d$ neighbours in $A$, or at least $d+1$ of them. In the first case, we argue that we may perform a good guessing step in the closed neighbourhood of $v$, because the degree of $v$ is bounded and when all the neighbours of $v$ are deleted (assigned to $Z$), then one may greedily assign $v$ to $A$. In the second case, we observe that we cannot assign too many such vertices $v$ to $A$, as otherwise we would obtain a subgraph of $G[A]$ with too high average degree. Let us now proceed to the formal arguments.

\begin{algrule}\label{rule:greedy}
Assume there exists a vertex $v \in Q$ such that $|N_G(v) \cap Q| < \degbound d$
and $|N_G(v) \cap A| \leq d$. Let $r = |N_G(v) \cap Q|$ and
$v_1,v_2,\ldots,v_r$ be an arbitrary ordering of the neighbours of $v$ in $Q$.
Let $\gamma = \gamma(r) \geq 1$ be such that
$$\gamma^{-1} + \gamma^{-2} + \ldots  + \gamma^{-r-1} = 1.$$
Randomly, make one of the following decisions:
\begin{enumerate}
\item for $1 \leq i \leq r$, with probability $\gamma^{-i}$ assign
$v_1,v_2,\ldots, v_{i-1}$ to $Z$ and $v_i$ to~$A$;
\item with probability $\gamma^{-r-1}$ assign all vertices $v_1,v_2,\ldots,v_r$ to $Z$
and $v$ to~$A$.
\end{enumerate}
\end{algrule}
Note that the choice of $\gamma$ not only ensures that the probabilities of the
options in Rule \ref{rule:greedy} sum up to one, but also that
$\gamma(r) \leq \gamma(\lceil \degbound d \rceil - 1 ) < 2$.
We now show a bound on the probability that an application of Rule \ref{rule:greedy} is correct.
\begin{lemma}\label{lem:greedy}
Assume that $A \subseteq X$ and $Z \cap X = \emptyset$ before Rule \ref{rule:greedy} is applied.
Then exactly one of the decisions considered in Rule \ref{rule:greedy} leads to a correct application.
Moreover, if in the correct decision exactly $i_0$ vertices are assigned to $A \cup Z$, then the probability of choosing the correct one is equal to $\gamma^{-i_0}$.
\end{lemma}
\begin{proof}
Firstly observe that the decisions in Rule \ref{rule:greedy} contradict each other, so at most one of them can lead to a correct application.

Assume that $(N_G(v) \cap Q) \cap X \neq \emptyset$ and let $v_{i_0}$ be the vertex from $(N_G(v) \cap Q) \cap X$ with the smallest index. Then the decision, which assigns all the vertices of $N_G(v) \cap Q$ with smaller indices to $Z$ and $v_{i_0}$ to $A$ leads to a correct application. Moreover, it assigns exactly $i_0$ vertices to $A\cup Z$ and the probability of choosing it is equal to $\gamma^{-i_0}$.

Assume now that $(N_G(v) \cap Q) \cap X=\emptyset$. We claim that $v\in X$. Assume otherwise; then $v$ has at most $d$ neighbours in $X$, so by Proposition~\ref{prop:greedy} after greedily incorporating it to $X$ we would still have $G[X]$ being a $d$-degenerate graph. This contradicts maximality of $X$. Hence, we infer that the decision which assigns all the neighbours of $v$ from $Q$ to $Z$ and $v$ itself to $A$ leads to a correct application, it assigns exactly $r+1$ vertices to $A\cup Z$ and has probability $\gamma^{-r-1}$.\maybeqed
\end{proof}

We now handle vertices with more than $d$ neighbours in $A$. Intuitively, there can be at most $d|A|$ such vertices assigned to $A$, as otherwise $A$ would have an induced subgraph with too high average degree. Hence, if there is significantly more than $2d|A|$ such vertices in total, then picking one of them at random with probability higher than $1/2$ gives a vertex that needs to be assigned to $Z$. Let us introduce a constant $\palantratio > 2$.
\begin{algrule}\label{rule:palanty}
If there are at least $\palantratio d|A|$ vertices in $Q$ that have more than $d$
neighbours in $A$, choose one such vertex uniformly at random
and assign it to $Z$.
\end{algrule}
\begin{lemma}\label{lem:palanty}
Assume that $A \subseteq X$ and $Z \cap X = \emptyset$ before Rule \ref{rule:palanty} is applied. Then the application of Rule \ref{rule:palanty}
is correct with probability at least $1-1/\palantratio$.
\end{lemma}
\begin{proof}
Let $P = \{v \in Q: |N_G(v) \cap A| > d\}$. As $|P|\geq \palantratio d|A|$, to prove the lemma it suffices
to show that $|P \cap X| < d|A|$. Assume otherwise, and consider the set
$((P \cap X) \cup A) \subseteq X$. The number of edges of
the subgraph of $G[X]$ induced by $(P \cap X) \cup A$ is at least
$$(d+1)|P \cap X| = d|P \cap X| + |P \cap X| \geq d(|P \cap X| + |A|) = d|(P \cap X) \cup A|.$$
This contradicts the assumption that $G[X]$ is $d$-degenerate, due to Proposition~\ref{prop:edgebound}.\maybeqed
\end{proof}
Note that $1 - 1/\palantratio > 1/2$ for $\palantratio > 2$.

We now show that if Rules \ref{rule:randedge}, \ref{rule:greedy} and \ref{rule:palanty}
are not applicable, then $|A \cup Z|$ is large, which means
that the algorithm has already made decisions about a significant fraction of the vertices of the graph.
\begin{lemma}\label{lem:finish}
If Rules \ref{rule:randedge}, \ref{rule:greedy} and \ref{rule:palanty}
are not applicable, then $|A \cup Z| > \stopratio n$
for some constant $\stopratio > 0$ that depends only on the constants $d$,
$\EVratio$, $\degbound$ and $\palantratio$.
\end{lemma}
\begin{proof}
As Rule \ref{rule:randedge} is not applicable, $Q$
contains at most $\frac{2\EVratio}{\degbound}|Q|$ vertices of degree
at least $\degbound d$ in $G[Q]$. As Rule \ref{rule:greedy} is not applicable,
the remaining vertices have more than $d$ neighbours in $A$.
As Rule \ref{rule:palanty} is not applicable, we have that
$$\frac{\degbound-2\EVratio}{\degbound} |Q| < \palantratio d |A| \leq \palantratio d |A \cup Z|.$$
As $Q = V(G) \setminus (A \cup Z)$, simple computations show that this is equivalent to
$$\frac{|A \cup Z|}{|V(G)|} > \left(\frac{\palantratio d \degbound}{\degbound - 2\EVratio} + 1\right)^{-1},$$
and the proof is finished.\maybeqed
\end{proof}
Lemma \ref{lem:finish} ensures that at this point the algorithm has already performed enough
steps to achieve the desired success probability. Therefore, we may finish by brute-force.
\begin{algrule}\label{rule:finish}
If $|A \cup Z| > \stopratio n$ for the constant $\stopratio$ given
by Lemma \ref{lem:finish}, for each $v \in Q$ independently,
assign $v$ to $A$ or $Z$ with probability $1/2$ each, and finish the algorithm
by outputting the set $A$ if it induces a $d$-degenerate graph, or reporting an error otherwise.
\end{algrule}
We now summarize the bound on the success probability.
\begin{lemma}\label{lem:summarize}
The algorithm outputs the set $X$ with probability at least
$$\max\left(\sqrt{\frac{3\EVratio}{\EVratio-1}}, \gamma(\lceil \degbound d \rceil - 1), \frac{\palantratio}{\palantratio-1}\right)^{-\stopratio n} 2^{-(1-\stopratio)n},$$
which is equal to $(2-\eps_d)^n$ for some $\eps_d > 0$.
\end{lemma}
\begin{proof}
Recall that $\frac{3\EVratio}{\EVratio-1} < 4$, $\gamma(\lceil \degbound d \rceil - 1) < 2$,
$\frac{\palantratio}{\palantratio-1} < 2$ and $\stopratio > 0$, by the choice of the constants
and by Lemma \ref{lem:finish}.
Therefore, it suffices to prove that, before Rule \ref{rule:finish} is applied,
the probability that $A \subseteq X$ and $Z \cap X = \emptyset$ is at least
$$\max\left(\sqrt{\frac{3\EVratio}{\EVratio-1}}, \gamma(\lceil \degbound d \rceil - 1), \frac{\palantratio}{\palantratio-1}\right)^{-|A \cup Z|}.$$
However, this is a straightforward corollary of Lemmata
\ref{lem:randedge}, \ref{lem:greedy} and \ref{lem:palanty}.\maybeqed
\end{proof}

This concludes the proof of Theorem \ref{thm:alg}.
In Table \ref{table:values} we provide a choice of values of the constants for small
values of $d$, together with corresponding value of $2-\eps_d$.

\begin{table}[htb]
\begin{center}
\begin{tabular}{|c|c|}
\hline
$d$ &  1 \\
$\EVratio$ &  4.0238224 \\
$\degbound$ &  $9$ \\
$\palantratio$ &  2.00197442 \\
$\stopratio$ &  0.050203 \\
$2-\eps_d$ &  1.99991 \\
\hline
$d$ &  2 \\
$\EVratio$ &  4.00009156 \\
$\degbound$ &  $17/2$ \\
$\palantratio$ &  2.00000763 \\
$\stopratio$ &  0.01449 \\
$2-\eps_d$ &  1.9999999 \\
\hline
$d$ &  3 \\
$\EVratio$ &  4.000000357628 \\
$\degbound$ &  $25/3$ \\
$\palantratio$ &  2.0000000298 \\
$\stopratio$ &  0.0066225 \\
$2-\eps_d$ &  1.9999999999 \\
\hline
  \end{tabular}
  \qquad
\begin{tabular}{|c|c|}
\hline
$d$ &  4 \\
$\EVratio$ &  4.000000001397 \\
$\degbound$ &  $33/4$ \\
$\palantratio$ &  2.0000000001164 \\
$\stopratio$ &  0.0037736 \\
$2-\eps_d$ &  1.9999999999996 \\
\hline
$d$ &  5 \\
$\EVratio$ &  4.000000000005457 \\
$\degbound$ &  $41/5$ \\
$\palantratio$ &  2.0000000000004548 \\
$\stopratio$ &  0.0024331 \\
$2-\eps_d$ &  1.999999999999999 \\
\hline
$d$ &  6 \\
$\EVratio$ &  4.000000000000021316 \\
$\degbound$ &  $49/6$ \\
$\palantratio$ &  2.0000000000000017833 \\
$\stopratio$ &  0.0016978 \\
$2-\eps_d$ &  1.999999999999999997 \\
\hline
\end{tabular}
\end{center}
\caption{Example values of the constants together with the corresponding success probability.}
\label{table:values}
\end{table}

\section{Conclusions}\label{sec:conc}

We have shown that the \vddeg{} problem can be solved in 
time $(2-\eps_d)^n n^{\Oo(1)}$ for any fixed $d \geq 1$.
There are two natural questions arising from our work.
First, can the algorithm be derandomized? Rules \ref{rule:greedy}
and \ref{rule:palanty} can be easily transformed into appropriate branching rules,
but we do not know how to handle Rule \ref{rule:randedge} without randomization.

Second, our constants $\eps_d$ are really tiny even for small values of $d$.
This is mainly caused by two facts: the gain over a straightforward brute-force
algorithm in Rule \ref{rule:greedy} is very small (i.e.,
$\gamma(\lfloor \degbound d \rfloor)$ is very close to $2$)
and the algorithm falls back to Rule \ref{rule:finish} after processing
only a tiny fraction $\stopratio$ of the entire graph.
Can the running time of the algorithm be significantly improved? Another interesting question would be to investigate, whether the \vddeg{} problem can be solved in time $(2-\eps)^n  n^{\Oo(1)}$ for some universal constant $\eps$ that is independent of $d$.

Apart from the above questions,
we would like to state here a significantly more challenging goal.
Let $\mathcal{G}$ be a polynomially recognizable graph class of bounded degeneracy
(i.e., there exists a constant $d$ such that each $G \in \mathcal{G}$ is $d$-degenerate).
Can the corresponding \textsc{Maximum Induced $\mathcal{G}$-Subgraph} problem
be solved in $(2-\eps_\mathcal{G})^n$ time for some constant
$\eps_\mathcal{G} > 0$ that depends only on the class $\mathcal{G}$?
Can we prove some meta-result for such type of problems?

Our Rules \ref{rule:randedge} and \ref{rule:palanty} are valid for any such class
$\mathcal{G}$; however, this is not true for the greedy step in Rule \ref{rule:greedy}.
In particular, we do not know how to handle the
\textsc{Maximum Induced $\mathcal{G}$-Subgraph} problem faster than $2^n$ even if the input is assumed to be $d$-degenerate.

\vskip 1.0cm

\noindent{\bf{Acknowledgements.}} We would like to thank Marek Cygan, Fedor V. Fomin and Pim van 't Hof for helpful discussions.

\bibliographystyle{splncs}
\bibliography{flat_d-degenerate-vd}

\begin{thebibliography}{10}

\bibitem{fgk:m-c-jacm}
Fomin, F.V., Grandoni, F., Kratsch, D.:
\newblock A measure \& conquer approach for the analysis of exact algorithms.
\newblock J. ACM \textbf{56}(5) (2009)  1--32

\bibitem{rooij:domsetesa09}
van Rooij, J.M.M., Nederlof, J., van Dijk, T.C.:
\newblock Inclusion/exclusion meets measure and conquer.
\newblock In Fiat, A., Sanders, P., eds.: ESA. Volume 5757 of Lecture Notes in
  Computer Science., Springer (2009)  554--565

\bibitem{nasz-tcs}
Cygan, M., Pilipczuk, M.:
\newblock Exact and approximate bandwidth.
\newblock Theor. Comput. Sci. \textbf{411}(40-42) (2010)  3701--3713

\bibitem{bjorklund-focs}
Bj{\"o}rklund, A.:
\newblock Determinant sums for undirected hamiltonicity.
\newblock In: 51th Annual IEEE Symposium on Foundations of Computer Science
  (FOCS), IEEE Computer Society (2010)  173--182

\bibitem{planar}
Fomin, F.V., Todinca, I., Villanger, Y.:
\newblock Exact algorithm for the maximum induced planar subgraph problem.
\newblock In Demetrescu, C., Halld{\'o}rsson, M.M., eds.: ESA. Volume 6942 of
  Lecture Notes in Computer Science., Springer (2011)  287--298

\bibitem{sched}
Cygan, M., Pilipczuk, M., Pilipczuk, M., Wojtaszczyk, J.O.:
\newblock Scheduling partially ordered jobs faster than $2^n$.
\newblock In Demetrescu, C., Halld{\'o}rsson, M.M., eds.: ESA. Volume 6942 of
  Lecture Notes in Computer Science., Springer (2011)  299--310

\bibitem{capdom}
Cygan, M., Pilipczuk, M., Wojtaszczyk, J.O.:
\newblock Capacitated domination faster than ${O}(2^n)$.
\newblock In Kaplan, H., ed.: SWAT. Volume 6139 of Lecture Notes in Computer
  Science., Springer (2010)  74--80

\bibitem{irredundant}
Binkele-Raible, D., Brankovic, L., Cygan, M., Fernau, H., Kneis, J., Kratsch,
  D., Langer, A., Liedloff, M., Pilipczuk, M., Rossmanith, P., Wojtaszczyk,
  J.O.:
\newblock Breaking the $2^n$-barrier for irredundance: Two lines of attack.
\newblock J. Discrete Algorithms \textbf{9}(3) (2011)  214--230

\bibitem{twoconn}
Cygan, M., Pilipczuk, M., Pilipczuk, M., Wojtaszczyk, J.O.:
\newblock Solving the 2-disjoint connected subgraphs problem faster than $2^n$.
\newblock In Fern{\'a}ndez-Baca, D., ed.: LATIN. Volume 7256 of Lecture Notes
  in Computer Science., Springer (2012)  195--206

\bibitem{exact-sfvs}
Fomin, F.V., Heggernes, P., Kratsch, D., Papadopoulos, C., Villanger, Y.:
\newblock Enumerating minimal subset feedback vertex sets.
\newblock In Dehne, F., Iacono, J., Sack, J.R., eds.: WADS. Volume 6844 of
  Lecture Notes in Computer Science., Springer (2011)  399--410

\bibitem{razgon-mif}
Razgon, I.:
\newblock Exact computation of maximum induced forest.
\newblock In Arge, L., Freivalds, R., eds.: SWAT. Volume 4059 of Lecture Notes
  in Computer Science., Springer (2006)  160--171

\bibitem{exact-fvs}
Fomin, F.V., Gaspers, S., Pyatkin, A.V., Razgon, I.:
\newblock On the minimum feedback vertex set problem: Exact and enumeration
  algorithms.
\newblock Algorithmica \textbf{52}(2) (2008)  293--307

\bibitem{seth}
Impagliazzo, R., Paturi, R.:
\newblock On the complexity of k-{SAT}.
\newblock J. Comput. Syst. Sci. \textbf{62}(2) (2001)  367--375

\bibitem{seth2}
Calabro, C., Impagliazzo, R., Paturi, R.:
\newblock The complexity of satisfiability of small depth circuits.
\newblock In Chen, J., Fomin, F.V., eds.: IWPEC. Volume 5917 of Lecture Notes
  in Computer Science., Springer (2009)  75--85

\bibitem{cut-and-count}
Cygan, M., Nederlof, J., Pilipczuk, M., Pilipczuk, M., van Rooij, J.M.M.,
  Wojtaszczyk, J.O.:
\newblock Solving connectivity problems parameterized by treewidth in single
  exponential time.
\newblock In Ostrovsky, R., ed.: FOCS, IEEE (2011)  150--159

\bibitem{treewidth-lower}
Lokshtanov, D., Marx, D., Saurabh, S.:
\newblock {Known Algorithms on Graphs of Bounded Treewidth are Probably
  Optimal}.
\newblock In: Proceedings of the Twenty-Second Annual ACM-SIAM Symposium on
  Discrete Algorithms (SODA). (2011)  777--789

\bibitem{patrascu}
P\u{a}tra\c{s}cu, M., Williams, R.:
\newblock On the possibility of faster {SAT} algorithms.
\newblock In: Proceedings of the Twenty-First Annual ACM-SIAM Symposium on
  Discrete Algorithms (SODA). (2010)  1065--1075

\bibitem{sodaomc}
Cygan, M., Dell, H., Lokshtanov, D., Marx, D., Nederlof, J., Okamoto, Y.,
  Paturi, R., Saurabh, S., Wahlstr{\"o}m, M.:
\newblock On problems as hard as {CNFSAT}.
\newblock CoRR \textbf{abs/1112.2275} (2011)

\bibitem{regular}
Gupta, S., Raman, V., Saurabh, S.:
\newblock Fast exponential algorithms for maximum {\it r}-regular induced
  subgraph problems.
\newblock In Arun-Kumar, S., Garg, N., eds.: FSTTCS. Volume 4337 of Lecture
  Notes in Computer Science., Springer (2006)  139--151

\bibitem{small-tw}
Fomin, F.V., Villanger, Y.:
\newblock Finding induced subgraphs via minimal triangulations.
\newblock In Marion, J.Y., Schwentick, T., eds.: STACS. Volume~5 of LIPIcs.,
  Schloss Dagstuhl - Leibniz-Zentrum fuer Informatik (2010)  383--394

\bibitem{colourable}
Angelsmark, O., Thapper, J.:
\newblock Partitioning based algorithms for some colouring problems.
\newblock In Hnich, B., Carlsson, M., Fages, F., Rossi, F., eds.: CSCLP. Volume
  3978 of Lecture Notes in Computer Science., Springer (2005)  44--58

\bibitem{bicliques}
Gaspers, S., Kratsch, D., Liedloff, M.:
\newblock On independent sets and bicliques in graphs.
\newblock In Broersma, H., Erlebach, T., Friedetzky, T., Paulusma, D., eds.:
  WG. Volume 5344 of Lecture Notes in Computer Science. (2008)  171--182

\bibitem{serge}
Gaspers, S.:
\newblock Exponential Time Algorithms: Structures, Measures, and Bounds.
\newblock Ph{D} {T}hesis, University of Bergen (2008)

\bibitem{bimber}
Cygan, M., Pilipczuk, M., Pilipczuk, M., Wojtaszczyk, J.O.:
\newblock Kernelization hardness of connectivity problems in {\it d}-degenerate
  graphs.
\newblock In Thilikos, D.M., ed.: WG. Volume 6410 of Lecture Notes in Computer
  Science. (2010)  147--158

\bibitem{minorddeg1}
Kostochka, A.V.:
\newblock Lower bound of the hadwiger number of graphs by their average degree.
\newblock Combinatorica \textbf{4}(4) (1984)  307--316

\bibitem{minorddeg2}
Thomason, A.:
\newblock An extremal function for contractions of graphs.
\newblock Math. Proc. Cambridge Philos. Soc. \textbf{95}(2) (1984)  261--265

\bibitem{minorddeg3}
Thomason, A.:
\newblock The extremal function for complete minors.
\newblock J. Comb. Theory, Ser. B \textbf{81}(2) (2001)  318--338

\end{thebibliography}

\end{document}